\documentclass[conference]{IEEEtran}
\IEEEoverridecommandlockouts

\usepackage{mathrsfs}
\usepackage{color}
\usepackage{amsfonts,amssymb}
\usepackage{amsfonts}
\usepackage{subfigure}
\usepackage{booktabs}

\usepackage{amsmath}
\usepackage{enumerate}
\usepackage[ruled]{algorithm2e}
\usepackage{algorithmic}
\usepackage{bm}
\usepackage{makecell}

\usepackage{mathtools}
\usepackage{bbm}
\usepackage{dsfont}
\usepackage{pifont}
\usepackage{amsthm}
\usepackage{cite}

\usepackage{enumitem}
\usepackage{balance}

\newtheorem{remark}{Remark}
\newtheorem{theorem}{Theorem}
\newtheorem{lemma}{Lemma}

\begin{document}

\title{Multi-period Optimal Control for Mobile Agents Considering State Unpredictability}
\author{Chendi Qu, Jianping He, and Jialun Li 
	\thanks{
	 The authors are with the Dept. of Automation, Shanghai Jiao Tong University, and Key Laboratory of System Control and Information Processing, Ministry of Education of China, Shanghai, China. E-mail address: \{qucd21, jphe, jialunli\}@sjtu.edu.cn. This work was supported by the National Natural Science Foundation of China 61973218 and 62103266.
	}
}

\maketitle

\begin{abstract}
The optimal control for mobile agents is an important and challenging issue. Recent work shows that using randomized mechanism in agents' control can make the state unpredictable, and thus improve the security of agents. However, the unpredictable design is only considered in  single period, which can lead to intolerable control performance in long time horizon. This paper aims at the trade-off between the control performance and state unpredictability of mobile agents in long time horizon. 
Utilizing random perturbations consistent with uniform distributions to maximize the attackers' prediction errors of future states, we formulate the problem as a multi-period convex stochastic optimization problem and solve it through dynamic programming. Specifically, we design the optimal control strategy considering both unconstrained and input constrained systems. The analytical iterative expressions of the control  are further provided. 
Simulation illustrates that the algorithm increases the prediction errors under Kalman filter while achieving the control performance requirements successfully.

\end{abstract}

\vspace{-5pt}
\section{Introduction}
Nowadays, as technologies of perception, localization and motion planning gradually mature, mobile agents such as autonomous vehicles have been widely applied in various fields \cite{raunholt2021towards, jiao2021multi}. 
However, traditional researches on optimal control and planning usually pursue the objectives of optimal time, shortest path or minimum energy consumption \cite{ connors2007analysis}. Since autonomous vehicles become more and more pervasive in the industry and military fields, how to protect their security while controlling have raised great concerns these years \cite{pasqualetti2013attack}, but still remains an unsolved problem \cite{bijani2014review}. 

In this paper, we focus on the security breach where an external attacker can predict the states of a mobile agent deployed in the physical environment. This security issue is of practical importance, since the historical trajectory of a vehicle is inevitably public to the attacker, who can leverage numerous estimation methods for interception or attack \cite{qu2022moving}. For instance, \cite{prevost2007extended} uses extended Kalman filter to estimate the states of a moving object and predict its trajectory by a UAV and \cite{schulz2018multiple} presents a multiple model unscented Kalman filter to predict multi-agent trajectory. 
Notice that for an attacker, an accurate prediction of the system state is essential to the latter attack. Therefore, it is necessary for a control system to protect its history trajectory for being unpredictable to attackers. Some studies have already been carried out in these years. \cite{tsiamis2019motion} introduces a novel coding scheme to protect the secrecy of the robot motion planning. \cite{hibbard2019unpredictable} considers an entropy maximization problem on a partially observable Markov decision process to decrease the predictability of the decision-maker's trajectory. A secure control method is shown in \cite{bianchin2019secure} to ensure the resilience to attacks. 

Nevertheless, to the best of our knowledge, few studies have yielded the unpredictable control of the mobile agents. In \cite{li2020unpredictable}, the author proves the unpredictability of the system is greatest when the inputs satisfy a uniform distribution. However, the proposed stochastic control method is a single-period optimization, which can only guarantee an unpredictable trajectory but omits the overall system performance in a long time horizon. Due to the addition of random disturbances, the final state of the system may not reach the expected target state. 

Inspired by the above discussion, we present a novel control strategy for mobile agents with linear systems to ensure both unpredictability and control performance. A stochastic perturbation is added to control inputs at each step in order to make the system states unpredictable and we translate the problem into a multi-period convex optimization problem. Moreover, simple input constraints are also considered in the problem. The main contributions are summarized as follows:
\begin{itemize}
\item Aiming at the trade-off between state unpredictability and the control performance in long time horizon, we formulate a multi-period stochastic optimization problem for linear systems and solve it through stochastic dynamic programming. 
\item We design an optimal control strategy for both unconstrained and input-constrained mobile agents to ensure the security and be unpredictable for external attackers while achieving the system performance requirements successfully. Performance showed by simulation demonstrates the effectiveness of the algorithm. 
\end{itemize}

The remainder of the paper is organized as follows. Section \ref{preliminary} describes the problem of interest. Section \ref{dpsolver} solves the optimal solution of the unconstrained problem with dynamic programming, while Sec. \ref{in-con} studies the problem with simple input constraints. Simulation results are shown in Sec. \ref{sim}, followed by conclusions and future directions in Sec. \ref{conc}. 

\vspace{-5pt}
\section{Problem Formulation and Preliminaries}\label{preliminary}
\subsection{Model Description}
\begin{itemize}[leftmargin=*]
\item System model of mobile agent
\end{itemize}

Consider a mobile agent with a linear dynamic model whose discrete form is
\vspace{-4pt}
\begin{equation}\vspace{-4pt}
{x}_{k+1} = A_k {x}_k + B_k {u}_k,
\end{equation}
where ${x}_k$ is the state vector of the mobile agent, ${u}_k$ is the control input and $A_k, B_k$ are $n \times n$ and $ n \times m$ matrices for any $k = 0, 1,\cdots, N -1$. 
The output model of the mobile agent is
\vspace{-3pt}
\begin{equation}\vspace{-3pt}
{y}_{k} = C {x}_k,
\end{equation}
where ${y}_k$ is the output information such as agent's positions, and we have $C = \begin{bmatrix} I_{q \times q} \quad \text{0}_{q \times (n-q)} \end{bmatrix}$.

\begin{itemize}[leftmargin=*]
\item Prediction model of attacker
\end{itemize}

Suppose there is an attacker having exact knowledge of the dynamic model $A_k, B_k$ and $C$ of the mobile agent. Through observing the output information, the attacker is able to predict the future trajectory of the agent with some data fusion methods. The prediction model is described as:
\vspace{-3pt}
\begin{equation}
\vspace{-3pt}
\hat{{y}}_{k+1|k} = C\hat{{x}}_{k+1|k} = C ( A_k \hat{{x}}_{k|k} + B_k \hat{{u}}_{k|k}),
\end{equation}
where $\hat{{y}}_{k+1|k}, \hat{{u}}_{k|k}$ are predictions to the system output and control input at time $k$ separately. $\hat{{x}}_{k|k}$ is the posterior estimate of ${x}_{k}$. In our assumption, $\hat{{x}}_{k|k}$ is an unbiased estimation. Denote the attacker's prediction error at time $k$ as:
\vspace{-3pt}
\begin{equation}
\vspace{-3pt}
{e}_{k+1|k} = {y}_{k+1} - \hat{{y}}_{k+1|k}.
\end{equation}

\subsection{Problem of Interest}

Consider a task of controlling the mobile agent from an initial state $x_0$ to a target state $x_N^o$. An external attacker is trying to predict the future state during this process and carry out an attack or interception. Our objective is to design a control policy ${u}_k$ to maximize the state unpredictability while achieving the control performance requirements in the long time horizon. 

In order to increase the unpredictability, a random disturbance term ${\delta}_k$ is added to the control input at each step. The random variable ${\delta}_k$ satisfies a distribution $f_k(y)$ with the expectation $\mathbb{E}({\delta}_k) = {0}$ and variance $\mathbb{D}({\delta}_{k,i}) = {\sigma}_{k,i}^2$, where $\delta_{k,i}$ and $\sigma_{k,i}$ are $i$ th component of $\delta_k$ and $\sigma_k$ separately. Note that each component of ${\delta}_k$ is independent. We have
\vspace{-3pt}
\begin{equation}\label{uk}
\vspace{-3pt}
{u}_k = {\mu}_k + {\delta}_k,
\end{equation}
which means ${u}_k$ is also a random variable and for all $k=0,1,\cdots,N-1,\,i=1,\cdots,m$
\begin{equation}
\mathbb{E}({u}_k) = {\mu}_k, \, \mathbb{D}({u}_{k,i}) = {\sigma}_{k,i}^2.
\end{equation}
In this paper, we define $\sigma_k^2=[\sigma_{k,1}^2,\cdots,\sigma_{k,m}^2]^T$ as a vector.

We now express the control objective in mathematical forms.  Giving a fixed terminal time $N$ and target state $x_N^o$, we use a linear quadratic objective function to represent the system performance, denoted as $J_c\{u_{0:N-1}\}$. The optimization problem is written as
\vspace{-3pt}
\begin{align}\label{con_obj}
 \min \,J_c\{u_{0:N-1}\} = &\, \, \mathbb{E} [(x_N - x_N^o)^T H (x_N - x_N^o)] \nonumber \\ & + \sum_{k = 0}^{N-1} \mathbb{E} (x_k^T Q_k x_k + u_k^T R_k u_k),
\vspace{-3pt}
\end{align}
where $u_{0:N-1}=\{u_0, u_1, \cdots, u_{N-1}\}$ is the set of control inputs, $H, Q_k$ are positive semi-definite matrices, and each $R_k$ is a positive definite matrix. Note that $J_c$ reflects both the deviation to the target state and the cost of control during the process.
 Since each control input $u_k$ in (\ref{uk}) is a random variable, the above optimization function is expressed in an expectation form.

On the other hand, we use the attacker's prediction error ${e}_{k+1|k}$ at each step $k$ to represent the performance of the unpredictability, which is denoted as $J_p\{k\}$. Notice that ${e}_{k+1|k}$ is also a random variable, so we define $J_p\{k\}$ with the expectation form:
\vspace{-3pt}
\begin{equation}\label{pre_obj}
\vspace{-3pt}
J_p\{k\}=\mathbb{E}(\Vert {e}_{k+1|k} \Vert ^2 _2), k=0, 1, \cdots N-1. 
\end{equation}
Regard (\ref{pre_obj}) as a measurement of the state unpredictability. According to \cite{li2020unpredictable}, we formulate the state unpredictability objective into a max-min optimization problem:
\vspace{-3pt}
\begin{equation}
\vspace{-3pt}\label{unp_obj}
\max_{f_k(y)} \min_{\hat{{u}}_{k|k}} J_p\{k\},
\end{equation}
where $J_p$ is first minimized with the attacker's optimal estimation $\hat{u}_{k|k}$ to the control input, and then maximized with ${\delta}_k$ generated from the optimal distribution $f_k(y)$.

Therefore, considering both of the control performance and the state unpredictability, our objective is to optimize both \eqref{con_obj} and \eqref{unp_obj}. In the next section, we will introduce one basic lemma and transform this multi-objective problem into a solvable single-objective problem.

\section{Problem Reformulation and the Optimal Solution}\label{dpsolver}

\subsection{Problem Reformulation}\label{problemof}
In this subsection, we will introduce a lemma from \cite{li2020unpredictable} and formulate an optimization problem combining \eqref{con_obj} and \eqref{unp_obj}. 
\begin{lemma}\label{lem1}
Consider the problem \eqref{unp_obj}. 
Suppose the maximum variance of the random perturbation ${\delta}_k$ is ${\sigma_k^2}$. Then, $f_k(y)$ is the optimal distribution in the sense of probability iff
\vspace{-3pt}
\begin{equation}\nonumber
\vspace{-3pt}
\mathbb{D}({\delta}_k) = {\sigma_k^2},
\end{equation}
and
\vspace{-5pt}
\begin{equation}\nonumber
\vspace{-5pt}
f_k^*(y) = \left\{
\begin{array}{ll}
\frac{1}{(2\sqrt{3})^m \prod_{i = 1}^m {\sigma}_{k,i}^2}, y_i \in [-\sqrt{3} {\sigma}_{k,i}^2, \sqrt{3} {\sigma}_{k,i}^2]. \\
\\
0, \mathrm{otherwise},
\end{array}
\right.
\end{equation}
where ${\sigma}_{k,i}$ is the component of ${\sigma}_{k}$, $i = 1,\cdots,m$.
\end{lemma}
\noindent The detailed proof process is in \cite{li2020unpredictable}.

Lemma \ref{lem1} shows that when the variance of the perturbation $\sigma_k^2$ reaches the maximum value, the expectation of the attacker's prediction error $J_p\{k\}$ is the largest. This conclusion is consistent with our intuitions, which means the larger the control variances are, the more difficult it is for an attacker to predict the future states accurately. Therefore, to maximize the state unpredictability is to maximize the variance $\sigma_k^2$, and using the multivariate uniform distribution as $f_k^*(y)$ can optimize the problem \eqref{unp_obj}. Thus, we define an unpredictable utility function as follows to help formulate our optimization problem. 

Unpredictable Utility Function:  Based on the above discussion, we design a utility function, denoted by $J_u\{k\}$, to represent the unpredictability at each period $k$, and combine \eqref{con_obj} and \eqref{unp_obj} into a single-objective optimization problem. There are two requirements for designing this utility function: 
\begin{itemize}
\item The utility function is negatively correlated with all the variances $\sigma_{k,i}^2$, i.e., $J_u\{k\}$ is a monotonically decreasing linear function of $\sigma_{k,i}^2$. When we minimize the value of the function $J_u\{k\}$, it needs to enlarge the variance, and then ensure the higher unpredictability. Thus, the trade-off between the system performance and the unpredictability is well formulated. 
\item To maintain the convexity of the optimization objective, each $J_u\{k\}$ needs to be convex with $\sigma_{k,i}^2$ for $k=0, \cdots, N-1$. Clearly, if the function is non-convex, the global optimal solution is hard to be obtained and not unique. Meanwhile, it may be taken at the infinite boundary when there are no constraints on the variable. 
\end{itemize}


Therefore, considering these two requirements, without loss of generality, the unpredictable utility function is defined as 
\vspace{-3pt}
\[\vspace{-3pt}
J_u\{k\}=\sum_{i = 1}^{m}\frac{1}{{\sigma_{k,i}}^2}.
\]
Then, we formulate a single-objective optimization problem: 
\begin{subequations}\label{J1}
\begin{eqnarray}
\mathbf{P}_1: & & \min_{\mu_k, \sigma_k} J = \lambda_1 \mathbb{E} [(x_N - x_N^o)^T H (x_N - x_N^o)]  \\\nonumber & & +\lambda_2 \sum_{k = 0}^{N-1} \mathbb{E} (x_k^T Q_k x_k + u_k^T R_k u_k) + \sum_{k = 0}^{N-1} \sum_{i = 1}^{m} \frac{\lambda_{3,k}}{{\sigma_{k,i}}^2} \\ &&
\mathrm{s.t.} ~~
(1), (5), (6),\,k = 0,1,\cdots,N-1,\nonumber
\end{eqnarray}
\end{subequations}
\noindent where $\lambda_1>0, \lambda_2>0, \lambda_{3,k}\geqslant 0$ are weights of each term respectively and $\sigma_{k,i}$ is the $i$th component of $\sigma_k$. The optimization variables are $\mu_k$ and $\sigma_k$ at each step. We can see that the third term in $J$ is $\sum_{k = 0}^{N-1} \lambda_{3,k} J_u\{k\}$, representing the unpredictability utility, and when the variance of control is larger, the function value is smaller. Moreover, the Bellman functions of this problem for all $k$ are convex, which will be shown in the next subsection.

In this way, we obtain a multi-period and single-object convex optimization problem. We will solve the Problem $\mathbf{P}_1$ in the next subsection.

\subsection{The Optimal Control Policy}
Since Problem $\mathbf{P}_1$ is multi-period and convex, we solve the optimal solution with dynamic programming \cite{bellman1966dynamic}.
\begin{theorem}\label{the-1}
The optimal solution of Problem $\mathbf{P}_1$ is given by
\vspace{-3pt}
\begin{equation}\label{sol_p1}\vspace{-3pt}
    \left\{
        \begin{array}{ll}
            \mu_{k} = -G_{k} x_{k} + M_{k} \\
            {\sigma_{k,i}}^2 = (\frac{\lambda_{3,k}}{P_{k,ii}})^{\frac{1}{2}},
        \end{array}
    \right.
\end{equation}
where for $k=0,1,\cdots,N-1$
\vspace{-3pt}
\begin{equation}\label{PGM}\vspace{-3pt}
\left\{
\begin{array}{ll}
 G_k = P_k^{-1}B_k^T J_{1,k+1}^T A_k \\
 M_k = \frac{1}{2} P_k^{-1}B_k^T J_{2,k+1}^T \\
 P_k = \lambda_2 R_k + B_k^T J_{1,k+1} B_k
\end{array}
\right.
\end{equation}
and
\vspace{-3pt}
\begin{equation}\label{lambda}
\vspace{-3pt}
\left\{
\begin{array}{ll}
 W_k = \lambda_2 Q_k + A_k^T J_{1,k+1} A_k ,\,Z_k = J_{2,k+1} A_k \\
  J_{1,k} = W_k - A_k^T J_{1, k+1} B_k G_k,\,J_{1,N} = \lambda_1 H \\
 J_{2,k} = Z_k - J_{2,k+1} B_k G_k,\,J_{2,N} = 2\lambda_1 {x_N^o}^T H .
\end{array}
\right.
\end{equation}
\end{theorem}
\begin{proof}
See the proof in Appendix.
\end{proof}{}

Once the parameters $G_k, M_k, P_k$ are obtained offline, we are able to calculate a sequence of $\sigma_k^2$ with Theorem \ref{the-1}. According to Lemma \ref{lem1}, the optimal distribution of the perturbation $\delta_k$ is the multivariate uniform distribution. Therefore, we generate the $\delta_k$ at each step from a uniform distribution $f_k(y)$ whose expectation is $0$ and variance is $\sigma_k^2$. We have 
\vspace{-3pt}
\begin{equation}\label{dist}
\vspace{-3pt}
\delta_{k,i} \sim \mathcal{U}[-(3 \sigma_{k,i}^2)^{\frac{1}{2}}, (3 \sigma_{k,i}^2)^{\frac{1}{2}}].
\end{equation}
Then, together with formula (1) and (5), we can get the expectation sequence $\mu_k$ and control inputs $u_k$.

The time complexity of the dynamic programming algorithm is $O(Nn^3)$, where $n$ is the dimension of state $x_k$ and $N$ is total control steps.

Note that the expectation sequence $\mu_k$ decides the direction of convergence of the system, while the variance $\sigma_k^2$ determines the unpredictability of the system state. In a traditional LQR problem without uncertainty, the control law is usually given by $u_k = -K_k x_k$. As for our algorithm, an additional parameter $M_k$ is used to adjust $\mu_k$ (and $u_k$) since the disturbance term $\delta_k$ is added. We can observe from the formula (16) that $\sigma_k^2$ is positively correlated with $\lambda_{3,k}$, which means the unpredictability of the system can be increased by enlarging the value of $\lambda_{3,k}$. We will detail the effect of weight parameters on control performance in the simulation section.


\section{Dealing with Simple Input-Constraints}\label{in-con}
In most practical situations, the systems contain multiple constraints. However, it is not easy to solve a multi-period linear quadratic optimal control problem with complex constraints. Inspired by \cite{mare2007solution}, we give an algorithm to solve the problem $\mathbf{P}_1$ with simple input-constraints through dynamic programming.

Consider a constrained problem
\vspace{-3pt}
\begin{subequations}\label{J2}
\begin{eqnarray}
\mathbf{P}_2: & & \min_{\mu_k, \sigma_k} J = \lambda_1 \mathbb{E} [(x_N - x_N^o)^T H (x_N - x_N^o)]  \\\nonumber & & +\lambda_2 \sum_{k = 0}^{N-1} \mathbb{E} (x_k^T Q_k x_k + u_k^T R_k u_k) + \sum_{k = 0}^{N-1} \sum_{i = 1}^{m} \frac{\lambda_{3,k}}{{\sigma_{k,i}}^2} \\ &&
\mathrm{s.t.} ~ -\overline{u} \leqslant u_k \leqslant \overline{u}, \\
&& ~~~~~~(1), (5), (6),\,k = 0,1,\cdots,N-1,\nonumber
\end{eqnarray}
\end{subequations}
\noindent where $\overline{u} > 0$ is the upper bound of $|u_k|$ at each step. From Theorem \ref{the-1}, we obtain 
\vspace{-3pt}
\begin{equation}
\vspace{-3pt}
\begin{array}{ll}
    \delta_k \sim \mathcal{U}[-\overline{\delta}_k,\overline{\delta}_k], \; \overline{\delta}_{k,i} = (3 {\sigma_{k,i}}^2)^{\frac{1}{2}}.
\end{array}
\end{equation}
Then we can set a conservative bound for each $\mu_k$ as
\vspace{-3pt}
\begin{equation}\label{bound}
\vspace{-3pt}
-\overline{\mu}_k \leqslant \mu_k \leqslant \overline{\mu}_k,\;\overline{\mu}_k = \overline{u} - \tau \overline{\delta}_k,
\end{equation}
where $\tau$ can determine how conservative the control is.
 In this way, if we set $\tau=1$, the control inputs $u_k$ will definitely satisfy the constraints (\ref{J2}b).

The following theorems provide a solution for $\mathbf{P}_2$.
\begin{theorem}\label{the-2}
The solution of $\mathbf{P}_2$ is given by:
\vspace{-3pt}
\begin{equation}\label{sol_p2}\vspace{-3pt}
    \left\{
        \begin{array}{ll}
            \mu_{k} = -\Tilde{G_{k}} x_{k} + \Tilde{M_{k}} \\
            {\sigma_{k,i}}^2 = (\frac{\lambda_{3,k}}{P_{k,ii}})^{\frac{1}{2}},
        \end{array}
    \right.
\end{equation}
where for $k=0,1,\cdots,N-1$, $\Tilde{G_{k}}$ and $\Tilde{M_{k}}$ are defined as:
\begin{equation}\label{GM}
\begin{aligned}
&\Tilde{G_{k}} = 
    \left\{
        \begin{array}{ll}
            G_k \qquad & if |-{G_{k}} x_{k} + {M_{k}}| \leqslant \overline{\mu}_k, \\
            0 \qquad & otherwise,
        \end{array}
    \right.  \\
&\Tilde{M_{k}} = 
    \left\{
        \begin{array}{ll}
            M_k \qquad & if |-{G_{k}} x_{k} + {M_{k}}| \leqslant \overline{\mu}_k, \\
            -\overline{\mu}_k \qquad & if -{G_{k}} x_{k} + {M_{k}} < -\overline{\mu}_k,  \\
            \overline{\mu}_k \qquad & if -{G_{k}} x_{k} + {M_{k}} > \overline{\mu}_k,
        \end{array}
    \right.
\end{aligned}
\end{equation}
and the calculation of $G_k,M_k,P_k$ is the same as Theorem \ref{the-1}.
\end{theorem}

\begin{remark}\label{rmk-1}
The process of using dynamic programming to solve $\mathbf{P}_2$ is similar to $\mathbf{P}_1$. The difference is that due to the constraints on $\mu_k$, at each step the new parameter $\Tilde{G_{k}}$ and $ \Tilde{M_{k}}$ need to satisfy the equation (\ref{GM}) according to three different situations, and $G_k, M_k$ are substituted into $\Tilde{G_{k}}, \Tilde{M_{k}}$ during the continue calculation. 
\end{remark}

\begin{remark}\label{rmk-2}
The problem of Theorem \ref{the-2} is that we cannot get $x_k$ when computing the parameters $\Tilde{G_{k}}$ and $\Tilde{M_{k}}$ off-line. Therefore, we need to traverse all the possibilities of the parameter pair $(\Tilde{G_{k}}, \Tilde{M_{k}})$. Note that there are 3 possible values for $(\Tilde{G_{k}}, \Tilde{M_{k}})$ at each step and $3^N$ possible values for $(\Tilde{G_{0}}, \Tilde{M_{0}})$. But for a certain initial state $x_0$, only one pair of $(\Tilde{G_{0}}, \Tilde{M_{0}})$ satisfies all constraints in the problem. We propose Theorem \ref{the-4} to help find the feasible parameters.
\end{remark}

\begin{theorem}\label{the-4}
Denote $\Tilde{G}_{0:N-1}=[\Tilde{G}_0,\cdots,\Tilde{G}_{N-1}]$, $ \Tilde{M}_{0:N-1}=[\Tilde{M}_0,\cdots,\Tilde{M}_{N-1}]$. For an initial state $x_0$, whether the parameter pair $(\Tilde{G}_{0:N-1}, \Tilde{M}_{0:N-1})$ are feasible with the constraints is determined by following $N$ inequalities.
\vspace{-3pt}
\begin{equation}\label{feasible}
\vspace{-3pt}
-\overline{\mu}_{k} \leqslant E_k x_0 + F_k \leqslant \overline{\mu}_k,\, k = 0,1,\cdots,N-1
\end{equation}
where
\vspace{-3pt}
\begin{equation}\nonumber
\vspace{-3pt}
\left\{
    \begin{array}{ll}
        E_k = - \Tilde{G_k} \Pi_{i = 0}^{k-1} (A_k-B_k\Tilde{G_i})  \\
        F_k = -\Tilde{G_k} \sum_{i = 0}^{k-1} [\Pi_{j = i}^{k-1} (A_j-B_j \Tilde{G_j})]B_i \Tilde{M_{i}} + \Tilde{M_k}.
    \end{array}
\right.
\end{equation}
\end{theorem}

\vspace{-3pt}
\section{Simulation}\label{sim}
\vspace{-3pt}
In this section, we conduct multiple simulations on our algorithm to show the performance and the unpredictability of the system control. 

In order to simplify the problem, we choose a single-in-single-out system. The dynamic equation is
\vspace{-3pt}
\begin{equation}\nonumber
\vspace{-3pt}
x_{k+1} = x_k + (x_k + u_k) \Delta t,
\end{equation}
where $\Delta t = T/N$ means to divide the total duration $T$ evenly into $N$ steps. Then we have $A = 1+\Delta t, B =1$ according to formula (1). Our control object is to maximize the unpredictability of the system state while meeting the performance requirement. Let $x_N^o=0$ and $H =1, Q=0, R = \Delta t$. The optimization function is described as
$\min_{\mu_k, \sigma_k} J_{1} = \lambda_1 \mathbb{E} (x_N^2) +  \lambda_2 \sum_{k = 0}^{N-1} \mathbb{E} (u_k^2\Delta t) + \sum_{k = 0}^{N-1} \frac{\lambda_{3}}{{\sigma_{k}}^2}$.

Set the initial state $x_0 = 20$ and $T = 10, N = 50$. Firstly, fix $\lambda_2 = 1, \lambda_3 = 0.5$ unchanged. With different $\lambda_1$, the state value $x_k$ and control input variance $\sigma^2_k$ at each step are shown in Fig.\ref{lam1}. We can see that $x_k$ has converged to the vicinity of the target state $x_N^o =0$ at almost $k=20$ and slightly fluctuates around $0$. The variance $\sigma_k$ decreases slowly until $N=40$ while drops rapidly to a small value in the last few steps. Note that $\lambda_1$ represents the weight of the deviation of the final state. When $\lambda_1$ is larger ($\lambda_1 = 15$), we find that $\sigma^2$ drops more and $\sigma_N^2$ is smaller, leading to a smaller deviation between $x_N$ and $0$. Consequently, if there is a high requirement for system performance, it will be necessary to increase $\lambda_1$.

\begin{figure}[t]
\vspace{-5pt}
\centering
\includegraphics[width=0.31\textwidth]{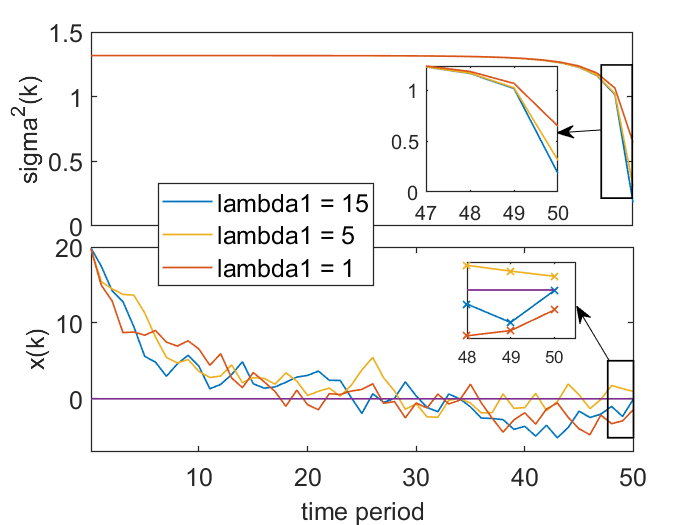}
\vspace{-5pt}
\caption{ Results with different $\lambda_1$.}
\label{lam1}
\vspace{-10pt}
\end{figure}

Now we fix $\lambda_1 = 5, \lambda_2 = 1$ and change $\lambda_3$. The curves of $x_k$ and $\sigma_N^2$ under different $\lambda_3$ are shown in Fig.\ref{lam3}. $\lambda_3$ represents the importance of the unpredictability and as $\lambda_3$ decreases from $1$ to $0.2$, the overall variance $\sigma^2$ gradually decreases too and the fluctuation range of $x_k$  becomes smaller.

\begin{figure}[t]
\centering
\includegraphics[width=0.31\textwidth]{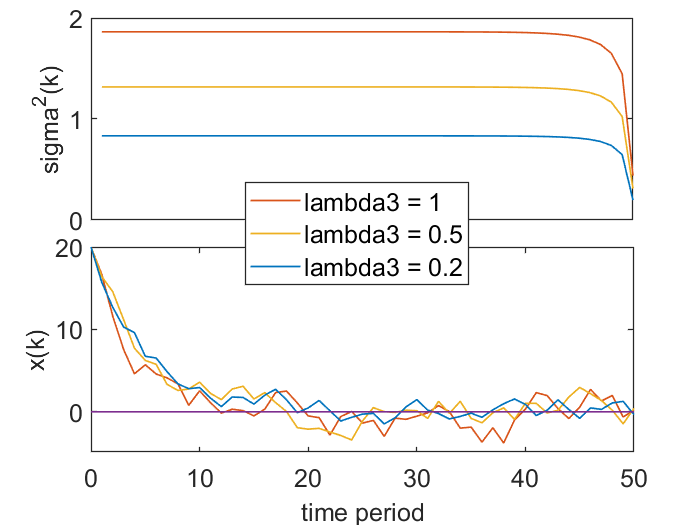}
\vspace{-5pt}
\caption{ Results with different $\lambda_3$.}
\label{lam3}
\vspace{-10pt}
\end{figure}

To demonstrate the unpredictability of our control law, we use another control method without random disturbance as a comparison (equivalent to $\lambda_3 = 0$), whose optimization object is given by $\min_{u_k} J_2 = x_N^2 + \sum_{k = 0}^{N-1} u_k^2 \Delta t$. Assume that the attacker has the optimal estimation of $u_k$, which means $\mathbb{E}(\hat{u}_{k|k}^*) = \mathbb{E}(u_k)$. We use Kalman Filter to do the one-step prediction. The observation noise is set to be $\mathcal{N}(0, 0.5)$ in the algorithm. At each step, the prediction of the next state is computed by $\hat{x}_{k+1|k} = A \hat{x}_{k|k} + B \hat{u}_{k|k}^*$. With two kinds of control methods, the prediction results and errors are illustrated in Fig.\ref{kf}. We can see that the prediction errors increase significantly after adding perturbation to the control inputs. The average and maximum prediction errors under different $\lambda_3$ are shown in Table.\ref{tab1} As $\lambda_3$ grows larger, the system state is more difficult to be predicted accurately.

\begin{figure}[t]
\centering
\includegraphics[width=0.31\textwidth]{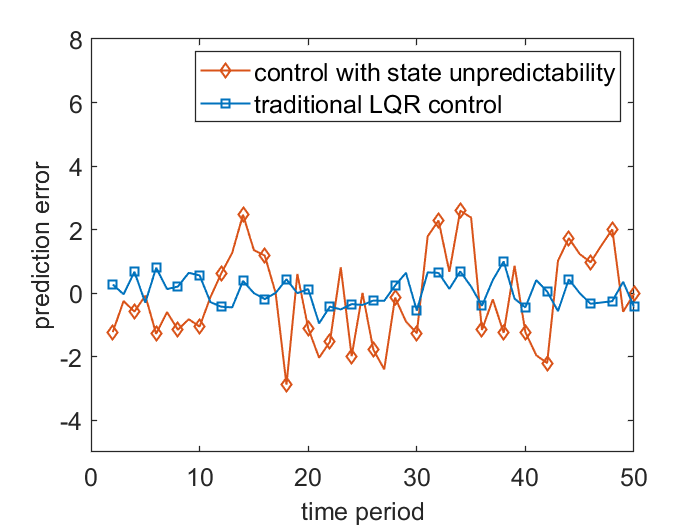}
\vspace{-5pt}
\caption{ Prediction errors with Kalman filter under different control policies. The blue line is traditional LQR control (($\lambda_3=0$), while the red line is our control policy with state unpredictability ($\lambda_3=0.5$).}
\label{kf}
\vspace{-10pt}
\end{figure}

\begin{table}\vspace{-20pt}
    \centering
    \caption{prediction error with Kalman filter ($\lambda_1 = \lambda_2 =1$)}
	\label{tab1}
    \begin{tabular}{ccccc}
    \toprule    
    $\lambda_3$ & 0 & 0.2 & 0.5 & 1  \\    
    \midrule   
    Ave. Error & 0.401 & 0.850 & 1.037 & 1.251 \\
    Max Error  & 1.389 & 2.440 & 2.795 & 3.124\\
    \bottomrule   
    \end{tabular}
\end{table}

The control results with and without simple input-constraints are shown in Fig.\ref{xu_cons}. Set the parameters $N=15$ and $\lambda_1 = 5, \lambda_2=1,\lambda_3=0.5$. The upper bound for control inputs $|u_k|$ is $\overline{u}=4$. We can find that our control policy with constraints has $|u_k|\leqslant4$ for all $k$, while the original policy without constraints exceeds the bound at $k=4$.
\begin{figure}[t]
\centering\includegraphics[width=0.31\textwidth]{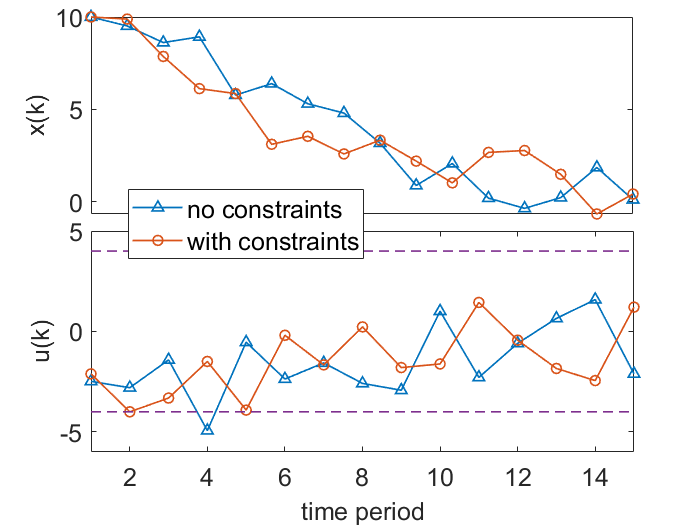}
\vspace{-5pt}
\caption{ Control with and without input-constraints.}
\label{xu_cons}
\vspace{-10pt}
\end{figure}

\vspace{-3pt}
\section{conclusion}\label{conc}\vspace{-3pt}
In this paper, an optimal control method considering state unpredictability for mobile agents is proposed. We add uniformly distributed random perturbations to the control inputs and formulate a multi-period convex optimization problem. The expectation and variance of the control inputs are solved through dynamic programming. The algorithm can also deal with the input-constrained systems. Our control method not only maximizes the attackers' prediction errors to the future states and guarantees the security of mobile agents, but also satisfies the system performance requirements. 

\vspace{-3pt}
\section*{Appendix}\vspace{-3pt}
Theorem \ref{the-1} is proved by mathematical induction.
According to the boundary conditions of the Bellman function, the value function at $N$ is written as:
\vspace{-3pt}
\begin{equation}\nonumber
\vspace{-3pt}
V(x_N) = \lambda_1 (x_N^T H x_N - 2q_N x_N),
\end{equation}
where $q_N = {x_N^o}^T H$. 
Let $J_{1,N} = \lambda_1 H$, $J_{2,N} = 2\lambda_1 q_N$
and $J_{3,N} = 0$. We have:
\vspace{-3pt}
\begin{equation}\label{vN}
\vspace{-3pt}
\begin{aligned}
V(x_N) =  x_N^T J_{1,N} x_N - J_{2,N} x_N + J_{3,N}.
\end{aligned}
\end{equation}
The Bellman function at time $N-1$ is:
\vspace{-3pt}
\begin{equation}\nonumber
\vspace{-3pt}
\small
\begin{aligned}
V(x_{N-1}) & = \min \{ \lambda_2 \mathbb{E} (x_{N-1}^T Q_{N-1} x_{N-1} \!+\! u_{N-1}^T R_{N-1} u_{N-1}) \\ & \quad + \sum_{i = 1}^{m} \frac{\lambda_{3,N-1}}{{\sigma_{N-1,i}}^2} + \mathbb{E}(V(x_{N})) \}.
\end{aligned}
\end{equation}
Substituting the dynamic model (1) into (\ref{vN}), we obtain:
\vspace{-3pt}
\begin{equation}\nonumber
\vspace{-3pt}
\begin{small}
\begin{aligned}
& \; \mathbb{E}(V(x_{N})) = \mathbb{E} (x_N^T J_{1,N} x_N) - J_{2,N} \mathbb{E}(x_N) + J_{3,N}.
\end{aligned}
\end{small}
\end{equation}
The function $V(x_{N-1})$ is simplified as:
\vspace{-3pt}
\begin{equation}\nonumber
\vspace{-3pt}
\begin{small}
\begin{aligned}
& V(x_{N-1}) = x_{N-1}^T W_{N-1} x_{N-1}\! -\! Z_{N-1} x_{N-1} \!+\! J_{3,N}+ \min\{\\
&  K_{N-1} \mu_{N-1} \!+\! \mu_{N-1}^T P_{N-1} \mu_{N-1} \!+\! \mathrm{Tr}(\Sigma_{N-1} P_{N-1}) \!+\!  \sum_{i = 1}^{m} \frac{\lambda_{3,N-1}}{{\sigma_{N-1,i}}^2} \},
\end{aligned}
\end{small}
\end{equation}
where
\vspace{-3pt}
\begin{equation}\nonumber
\vspace{-3pt}
\small
\left\{
\begin{array}{ll}
 W_k = \lambda_2 Q_k + A_k^T J_{1,k+1} A_k \\
 Z_k = J_{2,k+1} A_k \\
 P_k = \lambda_2 R_k + B_k^T J_{1,k+1} B_k \\
 K_k = 2 x_k^T A_k^T J_{1,k+1} B_k - J_{2,k+1}B_k,
\end{array}
\right.
\end{equation}
and $\Sigma_{N-1}$ is the covariance matrix of $\delta_{N-1}$. By differentiating $\mu_{N-1}$ and $\sigma_{N-1,i}$ separately, we have
\vspace{-3pt}
\begin{equation}\nonumber
\vspace{-3pt}
\begin{aligned}
&\frac{\partial {V}(x_{N-1})}{\partial \mu_{N-1}} = 2 \mu_{N-1}^T P_{N-1} + K_{N-1} = 0  \\
&\frac{\partial {V}(x_{N-1})}{\partial \sigma_{N-1,i}} = 2 P_{N-1,ii} \sigma_{N-1,i} - 2\frac{\lambda_{3,N-1}}{{\sigma_{N-1,i}}^3} = 0.
\end{aligned}
\end{equation}
Therefore the global optimal solution is:
\vspace{-3pt}
\begin{equation}
\vspace{-3pt}
\left\{
\begin{array}{ll}
 \mu_{N-1} = -G_{N-1} x_{N-1} + M_{N-1} \\
 {\sigma_{N-1,i}}^2 = (\frac{\lambda_{3,N-1}}{P_{N-1,ii}})^{\frac{1}{2}},
\end{array}
\right.
\end{equation}
where
\vspace{-3pt}
\begin{equation}\nonumber
\vspace{-3pt}
\left\{
\begin{array}{ll}
 G_k = P_k^{-1}B_k^T J_{1,k+1}^T A_k \\
 M_k = \frac{1}{2} P_k^{-1}B_k^T J_{2,k+1}^T,
\end{array}
\right.
\end{equation}
Substitute $\mu_{N-1}$ and $\sigma_{N-1}$ into $V(x_{N-1})$.
Let
\vspace{-3pt}
\begin{equation}\nonumber
\vspace{-3pt}
\small
\left\{
\begin{array}{ll}
 J_{1,k} = W_k - A_k^T J_{1, k+1} B_k G_k \\
 J_{2,k} = Z_k - J_{2,k+1} B_k G_k  \\
 J_{3,k} = - \frac{1}{4} J_{2,k+1} B_k M_k + 2 \sum_{i = 1}^{m} (\lambda_{3,k}P_{k,ii})^{\frac{1}{2}}+ J_{3, k+1},
\end{array}
\right.
\end{equation}
then we have
\vspace{-3pt}
\begin{equation}\nonumber
\vspace{-3pt}
V(x_{N-1}) = x_{N-1}^T J_{1,N-1} x_{N-1} \!-\! J_{2, N-1} x_{N-1} \!+\! J_{3, N-1}.
\end{equation}
Continue the above process for $k=N-2,\cdots,0$, then the mathematical induction is done.

\balance
\bibliographystyle{IEEEtran}
\bibliography{reference}

\end{document}